\documentclass[letterpaper, 10 pt, conference]{ieeeconf}  

\usepackage{amsmath,amssymb,amsfonts}
\usepackage{mathrsfs}
\usepackage{bm}
\usepackage{cite}
\usepackage[noend]{algpseudocode}
\usepackage{algorithmicx,algorithm}
\usepackage{graphicx}
\usepackage{subfig}
\usepackage{caption}

\usepackage{hyperref}
\usepackage{mathtools, cuted}
\usepackage{stfloats}
\usepackage{cleveref}
\usepackage{balance}
\usepackage{enumerate}
\usepackage{amsthm}
\usepackage{color}
\usepackage{booktabs}
\usepackage{lipsum}
\usepackage{setspace}

\newcommand{\RNum}[1]{\uppercase\expandafter{\romannumeral #1\relax}}
\newtheorem{definition}{Definition}
\newtheorem{lemma}{Lemma}

\newtheorem{theorem}{Theorem}
\newtheorem{corollary}{Corollary}
\newtheorem{assumption}{Assumption}
\newtheorem{problem}{Problem}
\newtheorem{proposition}{Proposition}

\newtheorem{remark}{Remark}

\usepackage{bm}
\usepackage{threeparttable}
\usepackage{mathrsfs} 
 
\IEEEoverridecommandlockouts                              
\title{\LARGE \bf
Primal-dual Estimator Learning: an Offline Constrained Moving Horizon Estimation Method with Feasibility and Near-optimality Guarantees
}

\author{Wenhan Cao, Jingliang Duan, Shengbo Eben Li, Chen Chen, Chang Liu, Yu Wang
\thanks{W. Cao, J. Duan, S. E. Li and C. Chen are with the School of Vehicle and Mobility, Tsinghua University, Beijing, 100084, China. {\tt\small Email: (cwh19, chenchen2020)@mails.tsinghua.edu.cn, duanjl15@163.com, lisb04@gmail.com}.}
\thanks{C. Liu is with the Department of Advanced Manufacturing and Robotics, Peking University, Beijing 100871, China. 
{\tt\small Email: changliucoe@pku.edu.cn}.}
\thanks{Y. Wang is with the Department of Electronic Engineering, Tsinghua University, Beijing, 100084, China. 
{\tt\small Email: yu-wang@tsinghua.edu.cn}.}
\thanks{Corresponding author: S. E. Li
}
}

\begin{document}

\maketitle
\thispagestyle{empty}
\pagestyle{empty}

\begin{abstract}
This paper proposes a primal-dual framework to learn a stable estimator for linear constrained estimation problems leveraging the moving horizon approach. To avoid the online computational burden in most existing methods, we learn a parameterized function offline to approximate the primal estimate. Meanwhile, a dual estimator is trained to check the suboptimality of the primal estimator during
execution time. Both the primal and dual estimators are learned from data using
supervised learning techniques, and the explicit sample size is provided, which enables us to guarantee the quality of each learned
estimator in terms of feasibility and optimality. This in
turn allows us to bound the probability of the learned estimator
being infeasible or suboptimal. Furthermore, we analyze the stability of the resulting estimator with a bounded error in the minimization of the cost function. Since our algorithm does
not require the solution of an optimization problem during runtime, state estimates can be generated online
almost instantly.
Simulation results are
presented to show the accuracy and time efficiency of the proposed framework compared to online optimization of 
moving horizon estimation and Kalman filter. To the best of our knowledge, this is the first learning-based state estimator with feasibility and near-optimality guarantees for linear constrained systems.
\end{abstract}
\setstretch{0.944}
\section{Introduction}
Estimating the state of a stochastic system is a long-lasting issue in the areas of engineering and science. It  draws much attention in different domains such as signal processing,  robotics, and econometrics \cite{musoff2009fundamentals}. 
For linear
systems, the Kalman filter gives the optimal estimate when
the process and the measurement noise obey Gaussian distributions \cite{KalmanFilter}. However, it is difficult to be applied in one typical case where states or disturbances are subjected to inequality constraints, especially for nonlinear constraints \cite{amor2018constrained}. Considering these constraints is crucial for bounded disturbances modeling, which will greatly facilitate the improvement of state estimation accuracy.


In contrast to the Kalman filter, moving horizon
estimation (MHE) offers the possibility
of incorporating constraints on the estimated systems \cite{allgower1999nonlinear,rao2003constrained,rawlings2012optimization}. At each instant, it is required to find a trajectory of state estimates online by solving a finite-horizon constrained optimization problem relying on recent measurements. 
It is shown that model predictive control (MPC) and MHE share symmetric structures \cite{rao2000moving}. This means that, similar to MPC, implementing MHE on fast dynamical systems with limited computation capacity remains generally challenging due to the heavy online computational burden.
To accelerate the online solving of MHE, a variety of fast optimization techniques have been proposed, including the interior-point nonlinear programming technique \cite{zavala2007fast,zavala2008fast,wachter2006implementation} and real-time iteration-based automatic code generation \cite{ferreau2012high}.

Compared with online MHE solvers, learning approximating MHE estimation laws offline can significantly improve the online estimation efficiency \cite{alessandri1999neural,alessandri2008moving,alessandri2011moving}. In particular, one can parameterize the MHE estimator using neural networks or a linear combination of basis functions, and then find a parameterized estimator that minimizes the MHE cost using supervised learning techniques. Some studies also apply reinforcement learning or variational inference to  obtain an offline estimator \cite{cao2021reinforced,li2020reinforcement, krishnan2015deep,karl2016deep}. However, existing offline estimation methods lack the ability to verify the estimation accuracy in real-time during the online application process. Nevertheless, in practice, it is critical to verify a state estimate
before it is utilized by a controller to ensure control performance. Besides, these methods also fail to certify the feasibility of the estimation law when considering constrained disturbances.


Inspired by the recently proposed primal-dual MPC framework \cite{zhang2019safe,zhang2020near}, this paper presents a primal-dual estimator learning method to learn an offline primal estimator with feasibility and stability guarantees, whose online optimality can be quantitatively evaluated in real-time using an offline dual estimator.
Specifically, our contributions can be summarized as follows:
\begin{enumerate}
\item Given a general constrained MHE problem, we establish the explicit form of its dual problem by introducing a minimum distance Euclidean projection
function.  Existing forms derived in \cite{goodwim2004duality,goodwin2005lagrangian} can be deemed as a special case of our setting, which considers both the discounted factor and disturbance constraints. 

\item In the offline phase, we employ a supervised learning scheme to train the primal and dual estimators and evaluate the feasibility and near-optimality of the trained estimator using a randomized verification methodology. Given an admissible probability of feasibility and suboptimality violation, the minimum sample sizes are provided for
the verification step. In the online phase, the primal estimator outputs a state estimate. In the meantime, we use the dual estimator to check the near-optimality of the current estimate using ideas from weak duality theory. If the check fails, we implement a backup estimator (such as an online MHE method) to guarantee the estimation accuracy. Therefore, in contrast to most existing offline methods \cite{alessandri1999neural,alessandri2008moving,alessandri2011moving,cao2021reinforced,krishnan2015deep,karl2016deep,li2020reinforcement}, our learning scheme guarantees the feasibility and near-optimality of the primal estimator.


\item Finally, we analyze the stability of the learned estimator, which shows that an upper bound of the state estimation error exists for any possible value of the estimator learning error under moderate assumptions.
\end{enumerate}

The remainder of this paper is organized as follows.  Section \ref{sec.II} presents the problem statement of the constrained MHE problem, and Section \ref{sec.III} derives the explicit form of the dual problem and formulates the primal and dual learning problems. Section \ref{sec.IV} proposes the algorithm for primal-dual learning to guarantee performance. Section \ref{sec.V} provides the stability analysis. Finally, we provide numerical results in Section \ref{sec.VI}  and draw conclusions in Section \ref{sec.VII}.

\textbf{Notation:} The Euclidean norm of the vector $x$ is denoted as ${\Vert x \Vert}_2$ and $x^\mathrm{T}Ax$ is denoted as ${\Vert x\Vert}^2_{A}$. A vector $x\geq0$ means that all the elements are greater than or equal to 0. We use $\mathbb{I}_{\Omega}$ to represent the integer lies in the set $\Omega$. For example, an integer $i\in\mathbb{I}_{[a,b]}$ represents $a\leq i\leq b $. $ \lambda_{\rm max} (P_{2}, P_{1})$ is the largest generalized
eigenvalue of $P_{2}$ and $P_{1}$. $I_{m \times m}$ represents the identity matrix.
\section{Problem Formulation}\label{sec.II}
This section formulates a constrained estimation problem using the moving horizon scheme.
We consider the stochastic system with process noise and measurement noise
\begin{equation}\label{eq.sys}
\begin{aligned}
x_{t+1}=A_tx_t+\xi_t\\
y_t=C_tx_t+\zeta_t,
\end{aligned}    
\end{equation}
where $x_t \in \mathbb{R}^n$ is the state, $y_t\in \mathbb{R}^m$ is the measurement, $\xi_{t}$ is the process noise, and $\zeta_{t}$ is the measurement noise. $\{\xi_t\}$ and $\{\zeta_t\}$ are both i.i.d sequences and independent of the initial state $x_0$. We suppose both the system noise and the measurement noise obey the truncated Gaussian distribution, i.e.,
\begin{equation}
\label{eq.truncated Gaussian}
\begin{aligned}
p(\xi_t)=\left\{
\begin{aligned}
&\frac{C_\xi}{\sqrt{(2\pi)^n|Q|}}e^{-\frac{1}{2}\xi_t^{\mathrm{T}}Q^{-1}\xi_t}\;&\xi_{t}\in {\Xi}_{\xi}
\\
&0&\xi_t\notin{\Xi}_{\xi}
\end{aligned}\right.\\
p(\zeta_t)=\left\{
\begin{aligned}
&\frac{C_\zeta}{\sqrt{(2\pi)^n|R|}}e^{-\frac{1}{2}\zeta_t^{\mathrm{T}}R^{-1}\zeta_t}\;&\zeta_t\in{\Xi}_{\zeta}
\\
&0&\zeta_t\notin{\Xi}_{\zeta}.
\end{aligned}
\right.
\end{aligned}
\end{equation}
Note that $C_\xi$ and $C_\zeta$ are the constant factors to normalize the probability density function and $Q,\;R\succ0$. The reason behind \eqref{eq.truncated Gaussian} is that an optimal estimator dealing with inequality constraints can be formulated under the assumption that the probability distributions are truncated Gaussian distributions \cite{lauvernet2009truncated}. 

A natural choice for the optimal estimate $\hat{x}_{t}^*$ is the most probable state ${x}_{t}$ given the measurement sequence $y_{1:t-1}$, which is known as the maximum a posteriori Bayesian estimation:
\begin{equation}\label{eq.maximum a posteriori Bayesian estimation}
\hat{x}_{1:t}^* = \arg\max_{{x}_{1:t}}{p\left({x}_{1:t}|y_{1:t-1}\right)}.
\end{equation}
This problem can be formulated as a quadratic optimization problem when applying the logarithm trick \cite{rao2003constrained}. However, this requires all the historical measurements to obtain the estimate, which is called full information estimation. This formulation is generally computationally intractable. To make the problem tractable, we need to bound the problem size. One strategy is to employ an MHE approximation which uses the most recent measurements to perform the estimation. The constrained MHE problem can be formulated as Problem \ref{problem.constrained MHE}. 

At time $t$, MHE considers the past measurements in a window of length $M_t\in\mathbb{I}_{[1,\infty)}$ and the past optimal estimate $\hat{x}_{t-M_t}^{*}$\footnote{This choice is typically called filtering prior.}. Thereby, the MHE optimizes over the initial estimate $\hat{x}_{t-M_t|t}$ and a sequence of $M_t$ estimates of the process noise $\hat{\xi}_{\cdot|t}=\{\hat{\xi}_{j|t}\}_{j=t-M_t}^{t-1}$. Combined, they define a sequence of $M_t+1$ state estimates $\hat{x}_{\cdot|t}=\{\hat{x}_{j|t}\}_{j=t-M_t}^{t}$\footnote{From the definition, $\hat{x}^*_{t|t} = \hat{x}^*_{t}$.} and a sequence of $M_t$ estimates of the measurement noise $\hat{\zeta}_{\cdot|t}=\{\hat{\zeta}_{j|t}\}_{j=t-M_t}^{t-1}$ through \eqref{eq.Constrained MHE problem(b)}.
\begin{problem}[Constrained MHE problem]
\label{problem.constrained MHE}
\begin{subequations}
\begin{equation}\label{eq.Constrained MHE problem(a)}
\begin{aligned}
\quad \quad \quad \min_{\hat{x}_{t-M_t|t},\hat{\xi}_{\cdot|t}}&V_{\rm MHE}(\hat{x}_{\cdot|t},\hat{\xi}_{\cdot|t}, \hat{\zeta}_{\cdot|t})
\end{aligned}    
\end{equation}
\begin{equation}\label{eq.Constrained MHE problem(b)}
\begin{aligned}
\text{\rm subject to}\qquad
\hat{x}_{i+1|t}&=A_i\hat{x}_{i|t}+\hat{\xi}_{i|t},
\\\hat{\zeta}_{i|t}&={y}_{i}-C_i\hat{x}_{i|t},\\
\end{aligned}    
\end{equation}
\begin{equation}\label{eq.Constrained MHE problem(c)}
\begin{aligned}
\qquad\qquad\hat{\xi}_{i|t}\in {\Xi}_{\xi}, \hat{\zeta}_{i|t}\in {\Xi}_{\zeta},i\in\mathbb{I}_{[t-M_t,t-1]}
\end{aligned}    
\end{equation}
\end{subequations}
where
\begin{equation}\label{eq.MHE cost}
\begin{aligned}
&V_{\rm MHE}(\hat{x}_{\cdot|t},\hat{\xi}_{\cdot|t}, \hat{\zeta}_{\cdot|t})
 =\gamma^{M_t} {\Vert \hat{x}_{t-M_t|t}-\hat{x}_{t-M_t}^{*} \Vert}^2_{P_{t-M_t}^{-1}}\\
&\quad+\sum_{i=t-M_t}^{t-1}\gamma^{t-i-1}{\Vert \hat{\xi}_{i|t} \Vert}^2_{Q^{-1}}+\sum_{i=t-M_t}^{t-1}\gamma^{t-i-1}{\Vert \hat{\zeta}_{i|t} \Vert}^2_{R^{-1}},
\end{aligned}
\end{equation}
where $\gamma\in[0,1)$ is the time-discounting factor which has
been previously suggested in \cite{knufer2018robust} to obtain stronger robustness bounds. ${\Vert \hat{x}_{t-M_t|t}-\hat{x}_{t-M_t}^{*} \Vert}^2_{P_{t-M_t}^{-1}}$ is the arrival cost which serves as an equivalent statistic
by penalizing the deviation of $\hat{x}_{t-M_t|t}$ away from $\hat{x}_{t-M_t}^*$. Besides, $P_{t-M_t}\succ0$ is the weighted matrix.
\end{problem}
\begin{remark} Generally, it is hard to obtain the analytic form of the arrival cost. Notable exceptions are unconstrained linear  problems, where $P_{t-M_t}$ can be obtained by solving the matrix Riccati equation:
\begin{equation}
\nonumber
\begin{aligned}
P_{i+1} = Q + A_iP_iA_i^{\mathrm{T}} - A_iP_iC_i^{\mathrm{T}}(R+C_iP_iC_i^{\mathrm{T}})^{-1}C_iP_iA_i^{\mathrm{T}}.
\end{aligned}    
\end{equation}
This results in a recurrent way to obtain the state estimation, which is equivalent to the well-known Kalman filter \cite{KalmanFilter,rao2000moving,rao2003constrained} when horizon length $M_t=1$.
\end{remark}
\begin{remark}
In this paper, we consider the prediction form of the estimation problem to simplify the notation. However, all the results can be directly extended to the filtering form.
\end{remark}

\begin{assumption}\label{assumption.convex set}
Both $\Xi_{\xi}$ and $\Xi_{\zeta}$ are convex sets.
\end{assumption}
\begin{proposition}
Problem \ref{problem.constrained MHE} is a convex optimization problem.
\begin{proof}
The objective function is quadratic and thus convex. The feasible region is also convex because the affine function can preserve convexity.
\end{proof}
\end{proposition}
\begin{assumption}
We assume that Problem \ref{problem.constrained MHE} is well-posed, i.e., a solution exists to Problem \ref{problem.constrained MHE} for $t\in\mathbb{I}_{[0,\infty)}$. The sufficient conditions for the existence of solutions 
are well studied in \cite{rao2003constrained}.
\end{assumption}


 
\section{Dual Problem and Estimator Approximation }\label{sec.III}
In this section, we first review some important conclusions about duality theory \cite{boyd2004convex} and then derive the dual problem of Problem \ref{problem.constrained MHE}. Finally, the supervised learning scheme is used to approximate the primal and dual estimators. 
\subsection{Duality theory}
Duality is often used in optimization to certify optimality of
a given solution. Consider the primal optimization problem:
\begin{equation}
\label{eq.primal optimization problem}
\begin{aligned}
\mathcal{P}:\;p^*=&\min F_0(x)\\
\text{subject to}\;\;F_i(x)&\leq0,\;\;i\in\mathbb{I}_{[0,p]} \\
H_i(x)&=0,\;\;i\in\mathbb{I}_{[0,q]}.
\end{aligned}
\end{equation}
The Lagrange function is defined as
\begin{equation}
L(x,v,\mu)=F_0(x)+\sum_{i=0}^{p}{v_iF_i(x)}+\sum_{i=0}^{q}{\mu_iH_i(x)}.
\end{equation}
Then, the corresponding Lagrange dual function is given by
\begin{equation}
g(v,\mu)=\inf_x L(x,v,\mu). 
\end{equation}
The Lagrange dual function gives us a lower bound on the optimal value $p^*$ of the primal problem \eqref{eq.primal optimization problem}. The calculation of the best lower bound leads to the Lagrange dual problem:
\begin{equation}
\label{eq.dual problem}
\begin{aligned}
\mathcal{D}:\;&d^*=\max g(v,\mu)\\
\text{subject to}\;\;&v\geq0.
\end{aligned}    
\end{equation}
The Lagrange dual problem $\mathcal{D}$ is a convex optimization problem regardless of whether the primal problem $\mathcal{P}$ is convex. It is well-known that $d^* \leq p^*$ always holds thanks to the weak duality theory and we refer to the difference $p^*-d^*$ as the duality gap.
\subsection{Dual problem of Constrained MHE}
From \eqref{eq.dual problem}, we establish the explicit form of the dual problem of Problem \ref{problem.constrained MHE}, which is given in Problem \ref{probelm.dual problem}. We defer detailed derivations to Appendix \ref{appendix.A}.
\begin{problem}[Duality of Problem \ref{problem.constrained MHE}]
\label{probelm.dual problem}
The dual problem of Problem \ref{problem.constrained MHE} is
\begin{subequations}
\begin{align}
&\qquad\begin{aligned}
\label{eq.Duality of Problem of MHE(a)}
\max_{\mu}\quad G(\lambda,\mu)
\end{aligned}\\
&\begin{aligned}
\label{eq.Duality of Problem of MHE(b)}
\text{\rm subject to}\;\;
&\lambda_{i-1}-A_i^\mathrm{T}\lambda_i-C_i^{\mathrm{T}}\mu_i=0\\
&\lambda_{t-1}=0, \; i\in\mathbb{I}_{[t-M_t+1,t-1]},
\end{aligned}    
\end{align}
\end{subequations}
where $\lambda_i$ and $\mu_i$ are Lagrange multipliers. Besides, $G(\lambda,\mu)$ is defined as
\begin{equation}\label{eq.Lagrange dual function}
\begin{aligned}
&G(\lambda,\mu):=-\frac{1}{4\gamma^{M_t}}{\Vert A_{t-M_t}^\mathrm{T}\lambda_{t-M_t}+C_{t-M_t}^\mathrm{T}\mu_{t-M_t} \Vert}^2_{P_{t-M_t}}\\
&+\sum_{i=t-M_t}^{t-1}\gamma^{t-i-1}{\Vert \Pi_{\Tilde{\Xi}_{\xi}}\left(\gamma^{i+1-t}Q^{1/2}\lambda_i\right) \Vert}^2_{2}\\
&+\sum_{i=t-M_t}^{t-1}\gamma^{t-i-1}{\Vert \Pi_{\Tilde{\Xi}_{\zeta}}\left(\gamma^{i+1-t}R^{1/2}\mu_i\right) \Vert}^2_{2}\\
&-\lambda_{t-M_t}^\mathrm{T}A_{t-M_t}\hat{x}_{t-M_t}^{*}-\mu_{t-M_t}^\mathrm{T}C_{t-M_t}\hat{x}_{t-M_t}^{*}+\sum_{i=t-M_t}^{t-1}\mu_{i}^{\mathrm{T}}y_{i}\\
&-\sum_{i=t-M_t}^{t-1}\lambda_{i}^{\mathrm{T}}Q^{1/2}\Pi_{\Tilde{\Xi}_{\xi}}\left(\gamma^{i+1-t}Q^{1/2}\lambda_i\right)\\
&-\sum_{i=t-M_t}^{t-1}\mu_{i}^{\mathrm{T}}R^{1/2}\Pi_{\Tilde{\Xi}_{\zeta}}\left(\gamma^{i+1-t}R^{1/2}\mu_i\right),
\end{aligned}
\end{equation}
$\Tilde{\Xi}_{\xi}$ and $\Tilde{\Xi}_{\zeta}$ are defined as
\begin{equation}\label{eq.definition of special sets}
\begin{aligned}
\Tilde{\Xi}_{\xi}:&=\left\{Q^{-1/2}x:x\in{\Xi}_{\xi}\right\}\\
\Tilde{\Xi}_{\zeta}:&=\left\{R^{-1/2}x:x\in{\Xi}_{\zeta}\right\},
\end{aligned}
\end{equation}
and $\Pi_{\Tilde{\Xi}_{\xi}}(\cdot)$ and $\Pi_{\Tilde{\Xi}_{\zeta}}(\cdot)$ are denoted as the minimum distance Euclidean projection onto the sets, i.e.,
\begin{equation}\label{eq.definition of projection}
\begin{aligned}
&\Pi_{\Tilde{\Xi}_{\xi}}(\cdot):\mathbb{R}^n\xrightarrow{}\mathbb{R}^n,\;
\Pi_{\Tilde{\Xi}_{\xi}}(z):=\arg\min_{x\in\Tilde{\Xi}_{\xi}}\left\{{\Vert x-\frac{1}{2}z \Vert}^2_2\right\}\\
&\Pi_{\Tilde{\Xi}_{\zeta}}(\cdot):\mathbb{R}^m\xrightarrow{}\mathbb{R}^m,\;
\Pi_{\Tilde{\Xi}_{\zeta}}(z):=\arg\min_{x\in\Tilde{\Xi}_{\zeta}}\left\{{\Vert x-\frac{1}{2}z \Vert}^2_2\right\}.
\end{aligned}
\end{equation}
\end{problem}
\subsection{Primal and Dual Learning Problems}
Given the formulation of Problems \ref{problem.constrained MHE} and \ref{probelm.dual problem}, we are now ready to train primal and dual estimators using supervised learning tools. 
We define all the information used to train the estimators as
\begin{equation}
\mathcal{I}_t:=\{y_{t-M_t:t-1},A_{t-M_t:t-1},C_{t-M_t:t-1},P_{t-M_t},\hat{x}_{t-M_t}^{*}\}.
\end{equation}
We observe that both the optimal estimator $\hat{X}^*(\mathcal{I}_t)$ (the optimal solution of Problem \ref{problem.constrained MHE}, i.e., $\hat{X}^*(\mathcal{I}_t)=\{\hat{x}^*_{t-M_t|t},\hat{\xi}^*_{\cdot|t}\}$) and the optimal dual estimator $\Lambda^*(\mathcal{I}_t)$ (the optimal solution of Problem \ref{probelm.dual problem}, i.e., $\Lambda^*(\mathcal{I}_t)=\{\mu^*\}$) are determined by $\mathcal{I}_t$.

Suppose the primal and dual estimators are parameterized by approximate functions $\hat{X}(\mathcal{I};\eta)$ and $\Lambda(\mathcal{I};\omega)$, respectively, where $\eta$ and $\omega$ are function parameters. Then the primal learning problem is given by
\begin{equation}\label{eq.primal learning}
\eta^*=\arg\min_{\eta}\sum_{i=1}^{N}{\mathcal{L}\left(\hat{X}(\mathcal{I}^i;\eta),\hat{X}^*(\mathcal{I}^i)\right) }.
\end{equation}
Similarly, the parameters of the dual estimator can be optimized by
\begin{equation}\label{eq.dual learning}
\omega^*=\arg\min_{\omega}\sum_{i=1}^{N}{\mathcal{L}\left(\Lambda(\mathcal{I}^i;\omega),\Lambda^*(\mathcal{I}^i)\right)}.
\end{equation}
Note that $\{\mathcal{I}^i, \hat{X}^*(\mathcal{I}^i), \Lambda^*(\mathcal{I}^i)\}$ represents the $i$-th sample, $N$ denotes the sample size and $\mathcal{L}$ is the loss function which can be chosen as different formulations such as the $L_2$ loss function.

\section{Primal-dual Estimator Learning}\label{sec.IV}

In this section, we show how the
parameterized estimators solved by \eqref{eq.primal learning} and \eqref{eq.dual learning} can be used to efficiently ensure the
feasibility and near-optimality of the estimator during runtime, inspired by \cite{zhang2019safe,zhang2020near}.

\subsection{Offline Training Performance Guarantees}

Given \eqref{eq.primal learning} and \eqref{eq.dual learning}, one natural question is how to verify the feasibility and near-optimality of the learned parameterized estimators. To answer this question, we first review a useful lemma in the field of statistical learning theory.
\begin{lemma}[Smallest Sample Size for Reliable Performance \cite{tempo1996probabilistic}]
\label{lemma.Smallest Sample Size for Reliable Performance}
Suppose $q$ is a random vector. Let $\{q^1,q^2,...,q^N\}$ represents N i.i.d. samples. Then $\hat{u}_N=\max_{i=1,2,..,N} u(q^i)$ represents the estimate of the worst-case performance function $u_{\rm max}:=\max_{q \in \mathcal{Q}}u(q)$, where $\mathcal{Q}$ denotes the sample space. The smallest sample size that guarantees
\begin{equation}
Prob\{u(q)>\hat{u}_N\}\leq \epsilon
\end{equation} with confidence at least $1-\beta$ is given by
\begin{equation}\label{eq.smallest sample size}
N\geq\frac{\ln{\frac{1}{\beta}}}{\ln{\frac{1}{1-\epsilon}}}.
\end{equation}
\end{lemma}

This lemma provides a powerful tool to test the performance of the primal and dual estimators using the collected finite samples. Specifically, given a desired
maximum suboptimality level, we can use this lemma to verify that
the approximated estimator satisfies this suboptimality
level with high probability. This comes with the following Theorem \ref{theorem.Offline Training With Performance Guarantee}.

\begin{theorem}[Offline Training With Performance Guarantee]\label{theorem.Offline Training With Performance Guarantee}
Suppose we have $N_p$ samples $\{\mathcal{I}^{i}, \hat{X}^*(\mathcal{I}^i)\}$ for primal estimator learning and $N_d$ samples $\{\mathcal{I}^{i}, \Lambda^*(\mathcal{I}^i)\}$ for dual estimator learning, where $N_p\geq \frac{\ln{\frac{1}{\beta_p}}}{\ln{\frac{1}{1-\epsilon_p}}}$ and $N_d\geq \frac{\ln{\frac{1}{\beta_d}}}{\ln{\frac{1}{1-\epsilon_d}}}$. Let $\epsilon_p, \epsilon_d \in [0,1)$ be admissible primal and dual violation probabilities, and let $0<\beta_p,\beta_d\ll1$ be desired confidence levels. The desired suboptimality level of the learned primal and dual estimators are denoted as $\Delta_p$ and $\Delta_d$, respectively. If\footnote{To simplify the notation, we use $V_{\rm MHE}(\hat{X}(\mathcal{I}; \eta))$ 
to represent the MHE cost. We imply that other variables in \eqref{eq.MHE cost} such as $\hat{\zeta}$ can be defined by \eqref{eq.Constrained MHE problem(b)} with no effort once $\hat{X}(\mathcal{I}; \eta)$ is obtained. Similar simplification suits for $V_{\rm MHE}(\hat{X}(\mathcal{I})), \; G(\Lambda(\mathcal{I}))$, and $G(\Lambda(\mathcal{I};\omega))$. 
}
\begin{equation}\label{eq.primal sample condition}
\begin{aligned}
V_{\rm MHE}(\hat{X}(\mathcal{I}^i; \eta^*))\leq V_{\rm MHE}(\hat{X}^*(\mathcal{I}^i)) + \Delta_p,\\
\hat{X}(\mathcal{I}^i; \eta^*)\;\text{satisfies} \;\eqref{eq.Constrained MHE problem(c)},\;\;i\in\mathbb{I}_{[1,N_p]}
\end{aligned}
\end{equation}
holds, then with confidence at least $1-\beta_p$ the following inequality holds
\begin{equation}
\begin{aligned}
Prob\Big\{V_{\rm MHE}(\hat{X}(\mathcal{I}; \eta^*))\leq V_{\rm MHE}(\hat{X}^*(\mathcal{I})) + \Delta_p,\\ 
\hat{X}(\mathcal{I}; 
\eta^*)\;\text{satisfies}\;\eqref{eq.Constrained MHE problem(c)} \Big\}\geq1-\epsilon_p.
\end{aligned}
\end{equation}
Similarly, if
\begin{equation}\label{eq.dual sample condition}
\begin{aligned}
G(\Lambda(\mathcal{I}^i;\omega^*)) \geq G(\Lambda^*(\mathcal{I}^i)) - \Delta_d,\;\;i\in\mathbb{I}_{[1,N_d]}
\end{aligned}
\end{equation}
holds, then with confidence at least $1-\beta_d$ the following inequality holds
\begin{equation}
\begin{aligned}
Prob\Big\{G(\Lambda(\mathcal{I};\omega^*)) \geq G(\Lambda^*(\mathcal{I})) - \Delta_d \Big\}\geq1-\epsilon_d.
\end{aligned}
\end{equation}
\end{theorem}

See Appendix \ref{appendix.B} for detailed proofs. Although we choose the confidence levels $\beta_p$ and $\beta_d$ as a small number ($<10^{-6}$), the minimum sample size would not explode due to the logarithm operator in \eqref{eq.smallest sample size}. Generally, given the required $\epsilon_{p/d}, \Delta_{p/d}$ and $\beta_{p/d}$, Theorem \ref{theorem.Offline Training With Performance Guarantee} provides an effective way to determine whether the learned estimator needs to be retrained.
\subsection{Online Application Performance Guarantees}
Although we have established probabilistic guarantees for the near-optimality of the learned estimator, there still remain some extreme cases where we may get a poor state estimate. To avoid such cases, we use the weak duality property to examine the learned estimator in real-time.

\begin{theorem}[Online Application With Performance Guarantee]\label{theorem.Online Application With Performance Guarantee} Assume $\hat{X}(\mathcal{I}; \eta^*)$ satisfies \eqref{eq.Constrained MHE problem(c)}, then
\begin{equation}\label{eq.Online Application With Performance Guarantee}
\begin{aligned}
V_{\rm MHE}(\hat{X}(\mathcal{I}; \eta^*))&-V_{\rm MHE}(\hat{X}^*(\mathcal{I}))\leq \\
&V_{\rm MHE}(\hat{X}(\mathcal{I}; \eta^*))-G(\Lambda(\mathcal{I};\omega^*)).
\end{aligned}
\end{equation}
\end{theorem}

\begin{proof}
This can be easily verified by weak duality $G(\Lambda(\mathcal{I};\omega^*)) \leq G(\Lambda^*(\mathcal{I})) \leq V_{\rm MHE}(\hat{X}^*(\mathcal{I}))$.
\end{proof}
We use Theorem \ref{theorem.Online Application With Performance Guarantee} in our framework as follows: Let
$\Delta$ be the desired maximum suboptimality level. During the online application process, for a 
given parameter $\eta^*$, if the right hand side of \eqref{eq.Online Application With Performance Guarantee} is smaller than $\Delta$, the performance gap between the learned primal estimator $\hat{X}(\mathcal{I}; \eta^*)$ and the optimal estimator $\hat{X}^*(\mathcal{I})$ can be bounded by $\Delta$. So we call this property as ``$\Delta$- suboptimality". If the learned primal estimator $\hat{X}(\mathcal{I}; \eta^*)$ is guaranteed to be at most $\Delta$-suboptimal, its output would be applied in real-time. However, if the right hand side of \eqref{eq.Online Application With Performance Guarantee} is larger
than the predetermined suboptimality level $\Delta$, then a backup estimator (such as an online MHE method) will be used to provide the state estimate at this instant.
The following corollary bounds the failure probability of the online application.
\begin{corollary}[Violation Probability]
Suppose $\Delta:=\Delta_p+\Delta_d+\Delta_{\rm gap}$, where $\Delta_{\rm gap}$ represents the maximum duality gap, i.e., $\Delta_{gap}=\max_{\mathcal{I}}\big\{{V_{\rm MHE}(\hat{X}^*(\mathcal{I}))-G(\Lambda^*(\mathcal{I}))}\big\}$. Under the assumptions in Theorem \ref{theorem.Offline Training With Performance Guarantee},  if \eqref{eq.primal sample condition} and \eqref{eq.dual sample condition} hold, then 
\begin{equation}
\nonumber
\begin{aligned}
Prob\Big\{&V_{\rm MHE}(\hat{X}(\mathcal{I}; \eta^*)) -G(\Lambda(\mathcal{I};\omega^*))
\leq \Delta,\\ 
& \qquad\qquad \hat{X}(\mathcal{I}; 
\eta^*)\;\text{satisfies}\;\eqref{eq.Constrained MHE problem(c)} \Big\}\geq1-(\epsilon_p + \epsilon_d)
\end{aligned}    
\end{equation}
holds with confidence at least $1-(\beta_p+\beta_d)$. In most cases, Problems \eqref{problem.constrained MHE} and \eqref{probelm.dual problem} satisfy the strong duality as long as the Slater condition holds \cite{boyd2004convex}, which leads to $\Delta_{\rm gap}=0$.
\end{corollary}
\begin{proof}
Using the union probability inequality $Prob\{A\cup B\}\leq Prob\{A\}+Prob\{B\}$
and the results in Theorem \ref{theorem.Offline Training With Performance Guarantee} , we can easily end this proof.
\end{proof}

\begin{remark}
Theorem \ref{theorem.Online Application With Performance Guarantee} already provides a ``hard" certificate to judge the performance $V_{\rm MHE}(\hat{X}(\mathcal{I}; \eta^*))$ of the learned primal estimator. Besides, we can bring some ideas from the control theory, such as the safety shield \cite{S.E.L.ReinforcementControl}, to ensure online feasibility. 
\end{remark}

\subsection{Primal-dual MHE}

Based on the above theoretical analysis, our offline MHE method building on the primal-dual estimator learning framework is shown in Algorithm \ref{alg:A}. We refer to this method as primal-dual MHE (PD-MHE).

\begin{algorithm}[!htb]
\caption{ Primal-dual MHE}\label{alg:A}
\hspace{0.01in} {\bf Input:}
confidence level $0<\beta\ll1$, violation probability $\epsilon>0$, and suboptimality level $\Delta>0$

\hspace{0.01in} {\bf Select:}
$\beta_p,\;\beta_d,\;\epsilon_p,\;\epsilon_d,$ such that $\beta_p+\beta_d=\beta$, $\epsilon_p+\epsilon_d=\epsilon$
\\
\\
\hspace{0.01in} \textcolor{blue}{\bf Offline Training}
\begin{algorithmic}[1]
\State Learn primal estimator $\hat{X}(\mathcal{I};\eta^*)$ as in \eqref{eq.primal learning}
\State Learn dual estimator $\Lambda(\mathcal{I};\omega^*)$ as in \eqref{eq.dual learning}
\State Validate $\hat{X}(\mathcal{I};\eta^*)$ and $\Lambda(\mathcal{I};\omega^*)$ using Theorem \ref{theorem.Offline Training With Performance Guarantee}
\If{\eqref{eq.primal sample condition} and \eqref{eq.dual sample condition} are satisfied and $\Delta_d+\Delta_p \le \Delta$} 
\State End training
\Else
\State Repeat
\EndIf
\end{algorithmic}
\hspace{0.01in} \textcolor{blue}{\bf Online Application} (for $t=M_t,M_t+1,M_t+2,...$)
\begin{algorithmic}[1]
\State Obtain $\mathcal{I}_{t}$
\If{$V_{\rm MHE}(\hat{X}(\mathcal{I}_t; \eta^*))-G(\Lambda(\mathcal{I}_t;\omega^*)\leq \Delta$; $\hat{X}(\mathcal{I}_t; \eta^*)$ satisfies \eqref{eq.Constrained MHE problem(c)}}
\State Apply $\hat{X}(\mathcal{I}_t; \eta^*)$ to obtain the estimate
\Else
\State Apply a backup estimator to obtain the estimate
\EndIf
\end{algorithmic}
\end{algorithm}

\section{Stability  Analysis}\label{sec.V}
In this section, we prove the stability of the offline primal estimator solved by \eqref{eq.primal learning}. Our analysis follows similar ideas in \cite{schiller2022lyapunov}, with suitable extensions to account for the ``$\Delta$-suboptimality" of the learned estimator. We begin with some useful definitions and lemmas.



\begin{definition}[Exponential $\delta$-\emph{IOSS} \cite{schiller2022lyapunov}] \label{def.exponential delta-IOSS}
The system has an exponential incremental input/output-to-state stability ($\delta$-\emph{IOSS}) property if there exists a quadratic $\delta$-\emph{IOSS} Lyapunov function $W_{\delta}$ and $\bar{P}_1, \bar{P}_2, \bar{Q}, \bar{R}\succ0$ such that
\begin{subequations}
\begin{equation}\label{eq.exponential delta-IOSS (a)}
{\Vert x-\Tilde{x} \Vert}^2_{\bar{P}_1} \leq W_{\delta}(x, \Tilde{x}) \leq 
{\Vert x-\Tilde{x} \Vert}^2_{\bar{P}_2},
\end{equation}
\begin{equation}\label{eq.exponential delta-IOSS (b)}
W_{\delta}(x^{+}, \Tilde{x}^{+}) \leq \gamma W_{\delta}(x, \Tilde{x}) + {\Vert \xi - \Tilde{\xi} \Vert}^2_{\bar{Q}}+{\Vert y - \Tilde{y} \Vert}^2_{\bar{R}}.
\end{equation}
\end{subequations}
Here $x^{+}$($\Tilde{x}^{+}$) represents the next state of $x$($\Tilde{x}$), $\xi$($\Tilde{\xi}$) represents the process noise, and $y$($\Tilde{y}$) represents the measurement.
\end{definition}

\begin{lemma}[Quadratic $\delta$-\emph{IOSS} Lyapunov function \cite{schiller2022lyapunov}]\label{lemma.Quadratic delta-IOSS Lyapunov function} 
The system \eqref{eq.sys} admits a Quadratic $\delta$-\emph{IOSS} Lyapunov function if there exists $\gamma \in [0,1)$ and symmetric matrices $\bar{P}, \bar{Q}, \bar{R}\succ0$ such that
\begin{equation}
\begin{aligned}
\begin{bmatrix}
A^\mathrm{T}\bar{P}A-\gamma\bar{P}-C^\mathrm{T}\bar{R}C& A^\mathrm{T}\bar{P}B-C^\mathrm{T}\bar{R}D
\\
B^\mathrm{T}\bar{P}A-D^\mathrm{T}\bar{R}C & B^\mathrm{T}\bar{P}B-\bar{Q}_1-D^\mathrm{T}\bar{R}D
\end{bmatrix}\preceq0
\end{aligned}    
\end{equation}
holds, then $W_\delta(x,\Tilde{x})$ is a $\delta$-\emph{IOSS} Lyapunov function that satisfies $\bar{P}_1=\bar{P}_2=\bar{P}$ in \eqref{eq.exponential delta-IOSS (a)}. Here, $B=\begin{bmatrix}I_{n\times n}& 0_{n \times m}\end{bmatrix}$, $D=\begin{bmatrix}0_{m \times n}& I_{m\times m}\end{bmatrix}$ ,and $\bar{Q}_1=\begin{bmatrix}\bar{Q}&0_{n \times m}\\0_{m \times n}&0_{m \times m}\end{bmatrix}$.
\end{lemma}
Before proposing the main theorem, we define the maximum of the largest  generalized
eigenvalue $\lambda_{\rm max}$ as
\begin{equation}
\begin{aligned}
\lambda_{\rm max} := \max\Big\{&\max_{i \in \mathbb{I}_{[M_t, \infty)}}\{\lambda_{\rm max}(P_{i}^{-1}, P_{i-M_t}^{-1})\},\\
&\qquad \max_{i \in \mathbb{I}_{[0, M_t-1]}}\{\lambda_{\rm max}(P_{i}^{-1},P_{0}^{-1})\}\Big\}.
\end{aligned}
\end{equation}
\begin{theorem}[Stability of the Primal Estimator]\label{theorem.stability}
The proposed estimator with ``$\Delta$-suboptimality" is a stable estimator if $M_t >- \frac{\ln{4\lambda_{\rm max}}}{\ln{\gamma}}, \lambda_{\rm max} \geq \frac{1}{4}$ and 
\begin{equation}
\begin{aligned}
&\begin{bmatrix}
\mathbb{M}_{11}& \mathbb{M}_{12}
\\
\mathbb{M}_{12}^\mathrm{T} &\mathbb{M}_{22}
\end{bmatrix}
\preceq0
\\
&\mathbb{M}_{11} = A^\mathrm{T} P_{t-M_t}^{-1}A-\gamma P_{t-M_t}^{-1}-2C^\mathrm{T}{R}^{-1}C\\
&\mathbb{M}_{12} = A^\mathrm{T}P_{t-M_t}^{-1}B-2C^\mathrm{T}{R}^{-1}D\\
&\mathbb{M}_{22} = B^{\mathrm{T}}P_{t-M_t}^{-1}B - \overline{Q^{-1}} - 2D^{\mathrm{T}}R^{-1}D,\;t\in \mathbb{I}_{[M_t,\infty)}.
\end{aligned}
\end{equation}
In particular, the state estimation error is bounded above by
\begin{equation}\label{eq.stability results}
\begin{aligned}
&{\Vert \hat{x}_{t}^{\Delta}- x_{t} \Vert}_{P_{t-M_t}^{-1}} \leq 2\sqrt{\rho}^{t}{\Vert \hat{x}_{0}-{x}_{0} \Vert}_{P_{0}^{-1}}+\sqrt{\frac{2\Delta}{1-\rho^{M_t}}}\\
&\quad +2\sqrt{\frac{1}{1-\sqrt{\rho}}}\max_{i\in\mathbb{I}_{[0,t-1]}}\left\{\sqrt[4]{\rho}^{i}{\Vert {\xi}_{t-i-1} \Vert}_{Q^{-1}}\right\},
\end{aligned}
\end{equation}
where $\hat{x}^{\Delta}_{t}$ denotes the estimate obtained by a $\Delta$-suboptimallity estimator and $\rho=(4\lambda_{\rm max})^{\frac{1}{M_t}}\gamma<1$. Besides, $\overline{Q^{-1}}=\begin{bmatrix}Q^{-1}&0_{n \times m}\\0_{m \times n}&0_{m \times m}\end{bmatrix}$.
\end{theorem}
This theorem shows that for the appropriate horizon length, the error sequence of the estimate generated by Algorithm \ref{alg:A} can be bounded by a function of the initial estimation error, the maximum norm of the process noise with time discounted, and the desired suboptimality level $\Delta$. Although this result does not restrict the range of $\Delta$, for a large $\Delta$,
such an upper bound becomes meaningless.

\begin{remark}
Compared to the original definition of robustly globally
exponentially stable (RGES) given in \cite{knufer2018robust, schiller2022lyapunov}, the derived error bound in \eqref{eq.stability results} includes an additional term to reveal the effect of estimator suboptimality.
\end{remark}
\section{Numerical Results}\label{sec.VI}
In this section, we use a simple example to illustrate the  performance of the proposed algorithm. We consider the following stochastic system
\begin{equation}
\begin{aligned}
x_{t+1} &= Ax_t+\xi_t
\\
y_t &= Cx_t + \zeta_t,
\end{aligned}
\end{equation}
where $x_{t}=[x_{t}^{(1)}, x_{t}^{(2)}]^\mathrm{T}\in\mathbb{R}^2$, $y_t\in\mathbb{R}$,
\begin{equation}
\begin{aligned}
A = \begin{bmatrix}
1&0.1\\
0&1
\end{bmatrix},\; 
C = \begin{bmatrix}
1&0
\end{bmatrix}.
\end{aligned}
\end{equation}
$\xi_t\sim\mathcal{N}(0,Q) \wedge \xi_t\geq0$ and $\zeta_t\sim\mathcal{N}(0,R) \wedge \zeta_t\leq0$ are the noise satisfying truncated Guassian distributions. The covariance matrix of the noise is set to
\begin{equation}
\begin{aligned}
Q = \begin{bmatrix}
0.1^{2}&0\\
0&0.1^{2}
\end{bmatrix},\;
R = 1.
\end{aligned}
\end{equation}
The projection functions defined in \eqref{eq.definition of projection} can be analytically expressed as
\begin{equation}
\nonumber
\Pi_{\Tilde{\Xi}_{\xi}}(z)=\left\{
\begin{aligned}
&0,\;\;&z^{(1)}\geq0,\;z^{(2)}&\geq0\\
&-z^{(1)}/2,\;\;&z^{(1)}\leq0,\;z^{(2)}&\geq0
\\
&-z^{(2)}/2,\;\;&z^{(1)}\geq0,\;z^{(2)}&\leq0\\
&{\Vert{z/2}\Vert}^2_2,\;\;&z^{(1)}\leq0,\;z^{(2)}&\leq0.
\end{aligned}
\right.
\end{equation}
and
\begin{equation}
\Pi_{\Tilde{\Xi}_{\zeta}}(z)=\left\{
\begin{aligned}
&0,\;\;&z\leq0\\
&z/2,\;\;&z\geq0.
\end{aligned}
\right.
\end{equation}

For the sake of comparison, we consider the performance
indices given by the root mean square error (RMSE) and asymptotic root mean square error (ARMSE) as in \cite{alessandri2008moving,alessandri2011moving}. To demonstrate the performance of Algorithm \ref{alg:A}, we take the Kalman filter and online MHE as baselines. We illustrate our proposed PD-MHE using a Deep Neural Network
function approximator ( 3 hidden layers and the number of neurons are [512, 512, 512] )  with Rectified Linear Unit. To solve
online MHE, we use CasADi \cite{Andersson2019}, the state-of-the-art
optimization problem solver. We set $M_t=10$ and performed 200 Monte-Carlo experiments for each method, and the results are given in Fig. \ref{fig.results} and Table \ref{tab:1}. The simulations are performed on a computer equipped with Intel i9-7980 XE processor and NVIDIA Titan XP GPU. 

\begin{figure}[!htb]
\centering
\includegraphics[width=0.75\linewidth]{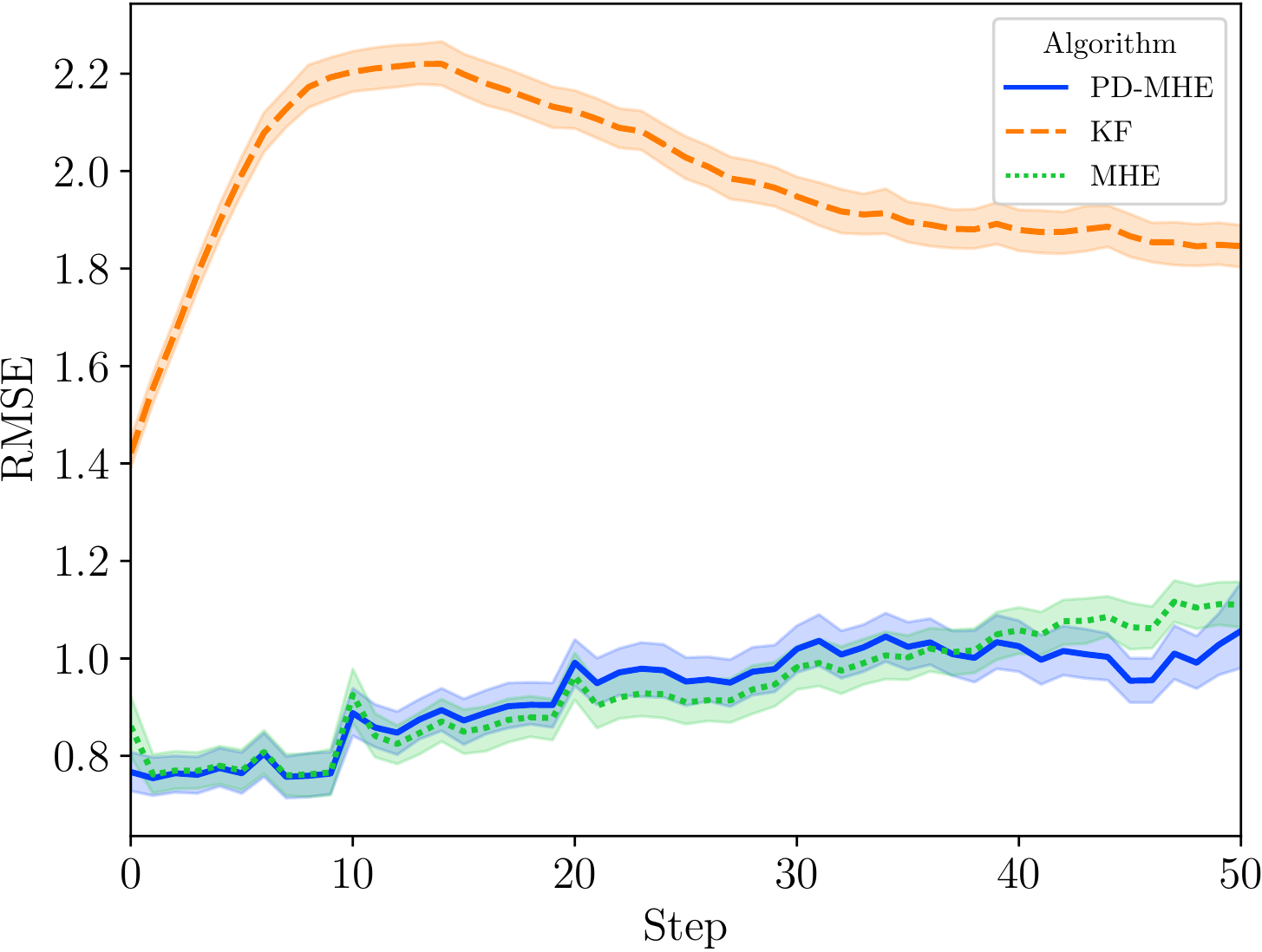}
\caption{Simulation results of the PD-MHE, KF, and online MHE. The
solid and dotted lines correspond to the means and the shaded
regions correspond to 95\% confidence intervals over 200
runs.}
\label{fig.results}
\end{figure}

\begin{table}[!htb]
    \caption{Comparsion of ARMSE and one-step computation time}
    \begin{center}
      \begin{tabular}{c c c}
    \hline
Algorithm&ARMSE& Computation Time  (ms)\\
\hline
 PD-MHE   & 0.9885    & 2.04 \\ 
\hline
 KF   &  1.9989   & 0.067 \\\hline
 MHE   & 0.9871    & 17.88 \\
 \hline
    \end{tabular} 
    \end{center}
    \label{tab:1}
\end{table}

From the results, we can see that the constraints defined on the process and measurement noise bring asymmetry to the probability density function, so that KF has already diverged in this constrained estimation problem.
Compared to KF, online MHE and PD-MHE capture the information given by the constraints, leading to better estimation accuracy. As for runtime, our algorithm allows us to obtain estimates significantly
faster than online MHE, with an average speed-up of over 8x compared to CasaADi. In summary, the proposed algorithm succeeds in learning a stable estimator for linear constrained systems, with negligible performance loss with respect to the online MHE.   
\section{Conclusions}\label{sec.VII}

In this paper, we proposed a new method, called primal-dual estimator learning, for approximating the explicit moving horizon estimation for linear constrained
systems. We approximated the moving horizon estimation
directly using supervised learning techniques, and invoked
two verification schemes to ensure the performance of
the approximated estimator. Since the proposed verification
scheme only requires the evaluation of primal and dual estimators, our
algorithm is computationally efficient, and can be implemented even on resource-constrained systems. The future work will
consider the iterative offline learning process and the influence of the capacity of the approximation function.

\appendix
\subsection{Derivation of Problem \ref{probelm.dual problem}}\label{appendix.A}
Considering the primal constrained moving horizon estimation problem \ref{problem.constrained MHE}, its Lagrange function is defined as
\begin{equation}
\begin{aligned}
&L(\hat x_{\cdot|t}, \hat{\xi}_{\cdot|t}, \hat{\zeta}_{\cdot|t},\lambda,\mu)
=\gamma^{M_t}{\Vert \hat{x}_{t-M_t|t}-\hat{x}_{t-M_t}^{*} \Vert}^2_{P_{t-M_t}^{-1}}\\
&\ +\sum_{i=t-M_t}^{t-1}\gamma^{t-i-1}{\Vert \hat{\xi}_{i|t} \Vert}^2_{Q^{-1}}+\sum_{i=t-M_t}^{t-1}\gamma^{t-i-1}{\Vert \hat{\zeta}_{i|t} \Vert}^2_{R^{-1}}
\\&\ +\sum_{i=t-M_t}^{t-1}\lambda_{i}^{\mathrm{T}}(\hat{x}_{i+1|t}-A_i\hat{x}_{i|t}-\hat{\xi}_{i|t})\\
&\ +\sum_{i=t-M_t}^{t-1}\mu_{i}^{\mathrm{T}}(y_{i}-C_i\hat{x}_{i|t}-\hat{\zeta}_{i|t}),
\end{aligned}    
\end{equation}
where ${\lambda_i}$ and $\mu_i$ are Lagrange multipliers. The Lagrange dual function is given by
\begin{equation}\label{eq.lagrange dual function(appendix)}
\begin{aligned}
g(\lambda,\mu)=\inf_{\hat x_{i|t}, \hat{\xi}_{i|t} \in \Xi_{\xi}, \hat{\zeta}_{i|t} \in \Xi_{\zeta}} L(\hat x_{\cdot|t}, \hat{\xi}_{\cdot|t}, \hat{\zeta}_{\cdot|t},\lambda,\mu).
\end{aligned}    
\end{equation}
Because the Lagrange function is convex with respect to $\hat x_{i|t},\; \hat{\xi}_{i|t}$, and $\hat{\zeta}_{i|t}$,  optimal variables $\hat x_{i|t}^*, \hat{\xi}_{i|t}^*, \hat{\zeta}_{i|t}^*$ can be calculated by the necessary condition:
\begin{subequations}
\begin{equation}\label{eq.derivative1}
\begin{aligned}
\frac{\partial{L(\cdot)}}{\hat{x}^*_{t-M_t|t}}&=2\gamma^{M_t}P^{-1}_{t-M_t}(\hat{x}^*_{t-M_t|t}-\hat{x}_{t-M_t}^{*})\\&-A_{t-M_t}^\mathrm{T}\lambda_{t-M_t}-C_{t-M_t}^\mathrm{T}\mu_{t-M_t}=0
\end{aligned}\\   
\end{equation}
\begin{equation}\label{eq.derivative2}
\begin{aligned}
\frac{\partial{L(\cdot)}}{\hat{x}^*_{i|t}}=\lambda_{i-1}-A_i^\mathrm{T}\lambda_i-C_i^{\mathrm{T}}\mu_i=0,\;i\in\mathbb{I}_{[t-M_t+1,t-1]}
\end{aligned}
\end{equation}
\begin{equation}\label{eq.derivative3}
\begin{aligned}
\frac{\partial{L(\cdot)}}{\hat{x}^*_{t|t}}=\lambda_{t-1}=0
\end{aligned}
\end{equation}
\begin{equation}\label{eq.derivative4}
\begin{aligned}
\hat{\xi}^*_{i|t}&=\arg\min_{\hat{\xi}_{i|t}}\left\{\gamma^{t-i-1}{\Vert \hat{\xi}_{i|t} \Vert}^2_{Q^{-1}}-\lambda_{i}^\mathrm{T}\hat{\xi}_{i|t}\right\}\\
&=\arg\min_{\hat{\xi}_{i|t}}\left\{{\Vert \hat{\xi}_{i|t} \Vert}^2_{Q^{-1}}-\gamma^{i+1-t}\lambda_{i}^\mathrm{T}\hat{\xi}_{i|t}\right\}
,\;i\in\mathbb{I}_{[t-M_t,t-1]}
\end{aligned}
\end{equation}
\begin{equation}\label{eq.derivative5}
\begin{aligned}
\hat{\zeta}^*_{i|t}=\arg\min_{\hat{\zeta}_{i|t}}\left\{{\Vert \hat{\zeta}_{i|t} \Vert}^2_{R^{-1}
}-\gamma^{i+1-t}\mu_{i}^\mathrm{T}\hat{\zeta}_{i|t}\right\},\;i\in\mathbb{I}_{[t-M_t,t-1]}.
\end{aligned}
\end{equation}
\end{subequations}
First, we observe that
\begin{equation}\label{eq.x_hat_{t-M_t}}
\hat{x}^*_{t-M_t|t}=\frac{1}{2\gamma^{M_t}}P_{t-M_t}(A_{t-M_t}^\mathrm{T}\lambda_{t-M_t}+C_{t-M_t}^\mathrm{T}\mu_{t-M_t})+\hat{x}_{t-M_t}^{*}.   
\end{equation}
Similar to the method proposed in \cite{goodwim2004duality,goodwin2005lagrangian,muller2006duality}, we express the \eqref{eq.derivative4} and \eqref{eq.derivative5} in the form of projection function. We denote $\Tilde{\xi}_{i|t}=Q^{-1/2}\hat{\xi}_{i|t}$ and
$\Tilde{\zeta}_{i|t}=R^{-1/2}\hat{\zeta}_{i|t}$. Then \eqref{eq.derivative4} and \eqref{eq.derivative5} can be rewritten as
\begin{equation}
\begin{aligned}
\Tilde{\xi}^*_{i|t}=\arg\min_{\Tilde{\xi}_{i|t}\in\Tilde{\Xi}_{\xi}}\left\{{\Vert \Tilde{\xi}_{i|t} \Vert}^2_2-\gamma^{i+1-t}\lambda_{i}^\mathrm{T}Q^{1/2}\Tilde{\xi}_{i|t}\right\},\\
\Tilde{\zeta}^*_{i|t}=\arg\min_{\Tilde{\zeta}_{i|t}\in\Tilde{\Xi}_{\zeta}}\left\{{\Vert \Tilde{\zeta}_{i|t} \Vert}^2_2-\gamma^{i+1-t}\mu_{i}^\mathrm{T}R^{1/2}\Tilde{\zeta}_{i|t}\right\},\\
i\in\mathbb{I}_{[t-M_t,t-1]}.
\end{aligned}    
\end{equation}
The solution can be formulated as
\begin{equation}\label{eq.xi_hat_{i}}
\begin{aligned}
\hat{\xi}^*_{i|t}=Q^{1/2}\Pi_{\Tilde{\Xi}_{\xi}}\left(\gamma^{i+1-t}Q^{1/2}\lambda_i\right),\;i\in\mathbb{I}_{[t-M_t,t-1]}\\
\hat{\zeta}^*_{i|t}=R^{1/2}\Pi_{\Tilde{\Xi}_{\zeta}}\left(\gamma^{i+1-t}R^{1/2}\mu_i\right),\;i\in\mathbb{I}_{[t-M_t,t-1]},
\end{aligned}     
\end{equation}
where $\Tilde{\Xi}_{\xi}$ and $\Tilde{\Xi}_{\zeta}$ are defined by \eqref{eq.definition of special sets}. At the meantime, $\Tilde{\Pi}_{\Xi_{\xi}}$ and $\Tilde{\Pi}_{\Xi_{\zeta}}$ are defined by \eqref{eq.definition of projection}. Plugging \eqref{eq.x_hat_{t-M_t}} and \eqref{eq.xi_hat_{i}} into \eqref{eq.lagrange dual function(appendix)}, we have \eqref{eq.Lagrange dual function}.

\subsection{Proof of Theorem \ref{theorem.Offline Training With Performance Guarantee}}\label{appendix.B}
\begin{proof}
We observe that the constraints defined by \eqref{eq.Constrained MHE problem(c)} can always be written as several inequalities according to the convex property by Assumption \ref{assumption.convex set} and the affine function defined by \eqref{eq.Constrained MHE problem(b)}. We denote them as $\mathcal{C}_i(\hat{X})\leq0, i\in\mathbb{I}_{[0,r]}$. Then for a given $\hat{X}(\cdot)$, consider the following function:
\begin{equation}
\begin{aligned}
&u(\mathcal{I}):=\max \big\{\max_{j}{\mathcal{C}_{j}(\hat{X}(\mathcal{I}))},\\
&\qquad \qquad V_{\rm MHE}(\hat{X}(\mathcal{I}))- V_{\rm MHE}(\hat{X}^*(\mathcal{I}))-\Delta_p
\big\}.
\end{aligned}    
\end{equation}
Define $\hat{u}_N:=\max_{i=1,...,N}{u(\mathcal{I}^{i})}$, where \{$\mathcal{I}^{i}$\} are independent samples. Following the result in Lemma \ref{lemma.Smallest Sample Size for Reliable Performance} and set $\hat{u}_N=0$, we can derive the results of the primal learning part in Theorem \ref{theorem.Offline Training With Performance Guarantee}. The proof of the dual learning part can be derived in a similar way. 
\end{proof}
\subsection{Proof of Theorem \ref{theorem.stability}}\label{appendix.C}
\begin{proof}
Based on Lemma \ref{lemma.Quadratic delta-IOSS Lyapunov function} and Definition \ref{def.exponential delta-IOSS}, we observe that ${\Vert \hat{x}^{\Delta}_{t} -{x}_{t}\Vert}^{2}_{\frac{1}{2}P_{t-M_t}^{-1}}$ is a $\delta$-\emph{IOSS} Lyapunov function which satisfies \footnote{In appendix \ref{appendix.C}, we use the notation $\Delta$ to represent the $\Delta$-suboptimality.}
\begin{equation}\label{eq.Lyapunov function}
\begin{aligned}
&{\Vert \hat{x}^{\Delta}_{t}- {x}_{t} \Vert}^{2}_{\frac{1}{2}P_{t-M_t}^{-1}} \leq \gamma{\Vert \hat{x}_{t-1|t}^{\Delta} - {x}_{t-1} \Vert}^{2}_{\frac{1}{2}P_{t-M_t}^{-1}}\\
&\qquad\qquad\quad +{\Vert \hat{\xi}_{t-1|t}^{\Delta} - {\xi}_{t-1} \Vert}^{2}_{\frac{1}{2}Q^{-1}}
+{\Vert \hat{\zeta}_{t-1|t}^{\Delta} \Vert}^{2}_{R^{-1}}.
\end{aligned}
\end{equation}

By applying \eqref{eq.Lyapunov function} $M_t$ times, we obtain
\begin{equation}
\begin{aligned}
&{\Vert \hat{x}_{t}^{\Delta}- x_{t} \Vert}^{2}_{\frac{1}{2}P_{t-M_t}^{-1}} \\
&\leq \gamma^{M_t}{\Vert \hat{x}_{t-M_t|t}^{\Delta} - {x}_{t-M_t} \Vert}^{2}_{\frac{1}{2}P_{t-M_t}^{-1}}\\
&\quad+\sum_{i=1}^{M_t}\gamma^{i-1}({\Vert \hat{\xi}_{t-i|t}^{\Delta} - {\xi}_{t-i} \Vert}^{2}_{\frac{1}{2}Q^{-1}}
+{\Vert \hat{\zeta}_{t-i|t}^{\Delta} \Vert}^{2}_{R^{-1}})\\
&\leq \sum_{i=1}^{M_t}\gamma^{i-1}({\Vert \hat{\xi}_{t-i|t}^{\Delta} \Vert}^{2}_{Q^{-1}}+
{\Vert {\xi}_{t-i} \Vert}^{2}_{Q^{-1}}
+{\Vert \hat{\zeta}_{t-i|t}^{\Delta} \Vert}^{2}_{R^{-1}})\\
&\quad+\gamma^{M_t}{\Vert \hat{x}_{t-M_t|t}^{\Delta}-\hat{x}_{t-M_t}^{\Delta} \Vert}^{2}_{P_{t-M_t}^{-1}}\\
&\quad+\gamma^{M_t}{\Vert \hat{x}_{t-M_t}^{\Delta}-{x}_{t-M_t} \Vert}^{2}_{P_{t-M_t}^{-1}}\\
&=V_{\rm MHE}(\hat{x}^{\Delta}_{\cdot|t},\hat{\xi}^{\Delta}_{\cdot|t}, \hat{\zeta}^{\Delta}_{\cdot|t})+\sum_{i=1}^{M_t}\gamma^{i-1}
{\Vert {\xi}_{t-i} \Vert}^{2}_{Q^{-1}}\\
&\quad+\gamma^{M_t}{\Vert \hat{x}_{t-M_t}^{\Delta}-{x}_{t-M_t} \Vert}^{2}_{P_{t-M_t}^{-1}}.
\end{aligned}
\end{equation}
According to the property of ``$\Delta$-suboptimality", $V_{\rm MHE}(\hat{x}^{\Delta}_{\cdot|t},\hat{\xi}^{\Delta}_{\cdot|t}, \hat{\zeta}^{\Delta}_{\cdot|t}) \leq V_{\rm MHE}(\hat{x}^{*}_{\cdot|t},\hat{\xi}^{*}_{\cdot|t}, \hat{\zeta}^{*}_{\cdot|t})+\Delta$. Upon the fact that the true underlining system is a  feasible solution, $V_{\rm MHE}({x}_{\cdot},{\xi}_{\cdot}, {\zeta}_{\cdot})$ is a trivial upper bound of $V_{\rm MHE}(\hat{x}^*_{\cdot|t},\hat{\xi}^*_{\cdot|t}, \hat{\zeta}^*_{\cdot|t})$:
\begin{equation}\label{eq.bound for lyapunov function}
\begin{aligned}
&{\Vert \hat{x}_{t}^{\Delta}- x_{t} \Vert}^{2}_{\frac{1}{2}P_{t-M_t}^{-1}} \\
&\leq V_{\rm MHE}({x}_{\cdot},{\xi}{\cdot}, {\zeta}_{\cdot})+\sum_{i=1}^{M_t}\gamma^{i-1}
{\Vert {\xi}_{t-i} \Vert}^{2}_{Q^{-1}}\\
&\quad+\gamma^{M_t}{\Vert \hat{x}_{t-M_t}^{\Delta}-{x}_{t-M_t} \Vert}^{2}_{P_{t-M_t}^{-1}} + \Delta \\
&=2\sum_{i=1}^{M_t}\gamma^{i-1}
{\Vert {\xi}_{t-i} \Vert}^{2}_{Q^{-1}}+2\gamma^{M_t}{\Vert \hat{x}_{t-M_t}^{\Delta}-{x}_{t-M_t} \Vert}^{2}_{P_{t-M_t}^{-1}} \\
& \quad + \Delta\\
&\leq2\gamma^{M_t}\lambda_{\rm max}(P_{t-M_t}^{-1},P_{t-2M_t}^{-1}){\Vert \hat{x}_{t-M_t}^{\Delta}-{x}_{t-M_t} \Vert}^{2}_{P_{t-2M_t}^{-1}} \\
&\quad+2\sum_{i=1}^{M_t}\gamma^{i-1}
{\Vert {\xi}_{t-i} \Vert}^{2}_{Q^{-1}} + \Delta.
\end{aligned}
\end{equation}
By assumption, we define $\rho^{M_t}:=4\lambda_{\rm max}\gamma^{M_t}<1$. Consider $t=kM_t+l, k\in\mathbb{I}_{[0,\infty)}, l\in\mathbb{I}_{[0, M_t-1]}$. Similar to \eqref{eq.bound for lyapunov function}, we can obtain
\begin{equation}\label{eq.error sequence (2)}
\begin{aligned}
&{\Vert \hat{x}_{l}^{\Delta} - x_{l} \Vert}^{2}_{\frac{1}{2}P_{0}^{-1}}\\
&\leq2\sum_{i=1}^{l}\gamma^{i-1}
{\Vert {\xi}_{t-kM_t-i} \Vert}^{2}_{Q^{-1}}+2\gamma^{l}{\Vert \hat{x}_{0}-{x}_{0} \Vert}^{2}_{P_{0}^{-1}}+\Delta.
\end{aligned}
\end{equation}
By applying \eqref{eq.bound for lyapunov function} $k$ times, we arrive at
\begin{equation}
\begin{aligned}
&{\Vert \hat{x}_{t}^{\Delta}- x_{t} \Vert}^{2}_{\frac{1}{2}P_{t-M_t}^{-1}} \\
&\leq \rho^{kM_t}{\Vert \hat{x}_{l}^{\Delta}- x_{l} \Vert}^{2}_{\frac{1}{2}P_{0}^{-1}}\\
&\quad+4\sum_{i=0}^{k-1}\rho^{iM_t}\sum_{j=1}^{M_t}{{\gamma}^{j-1}\Vert {\xi}_{t-iM_t-j} \Vert}^{2}_{\frac{1}{2}Q^{-1}}+\Delta\sum_{i=0}^{k-1}\rho^{iM_t}.
\end{aligned}
\end{equation}
Plugging \eqref{eq.error sequence (2)} into \eqref{eq.error sequence (1)}, we have
\begin{equation}\label{eq.error sequence (1)}
\begin{aligned}
&{\Vert \hat{x}_{t}^{\Delta}- x_{t} \Vert}^{2}_{\frac{1}{2}P_{t-M_t}^{-1}}\\
&\leq \rho^{kM_t}\Big(4\sum_{i=1}^{l}\gamma^{i-1}
{\Vert {\xi}_{t-kM_t-i} \Vert}^{2}_{\frac{1}{2}Q^{-1}}\\
&\qquad\qquad\qquad\qquad\qquad+4\gamma^{l}{\Vert \hat{x}_{0}-{x}_{0} \Vert}^{2}_{\frac{1}{2}P_{0}^{-1}}+\Delta \Big)\\
&\quad+4\sum_{i=0}^{k-1}\rho^{iM_t}\sum_{j=1}^{M_t}{\gamma}^{j-1}{\Vert {\xi}_{t-iM_t-j} \Vert}^{2}_{\frac{1}{2}Q^{-1}}+\Delta\sum_{i=0}^{k-1}\rho^{iM_t}.
\end{aligned}
\end{equation}
By assumption, we have $\rho\ge\eta$, thus
\begin{equation}
\begin{aligned}
&{\Vert \hat{x}_{t}^{\Delta}- x_{t} \Vert}^{2}_{\frac{1}{2}P_{t-M_t}^{-1}}\\
& \leq 4\sum_{i=1}^{l}\rho^{kM_t+i-1}
{\Vert {\xi}_{t-kM_t-i} \Vert}^{2}_{\frac{1}{2}Q^{-1}}+4\rho^{t}{\Vert \hat{x}_{0}-{x}_{0} \Vert}^{2}_{\frac{1}{2}P_{0}^{-1}}\\
&\quad+4\sum_{i=0}^{k-1}\sum_{j=1}^{M_t}\rho^{iM_t+j-1}{\Vert {\xi}_{t-iM_t-j} \Vert}^{2}_{\frac{1}{2}Q^{-1}}+\Delta\sum_{i=0}^{k}\rho^{iM_t}\\
&\leq 2\rho^{t}{\Vert \hat{x}_{0}-{x}_{0} \Vert}^{2}_{P_{0}^{-1}}+2\sum_{i=0}^{t-1}\rho^{i}{\Vert {\xi}_{t-i-1} \Vert}^{2}_{Q^{-1}}
+\Delta\sum_{i=0}^{k}\rho^{iM_t}\\
&\leq 2\rho^{t}{\Vert \hat{x}_{0}-{x}_{0} \Vert}^{2}_{P_{0}^{-1}}+2\sum_{i=0}^{t-1}\sqrt{\rho}^{i}\max_{i\in\mathbb{I}_{[0,t-1]}}\left\{\sqrt{\rho}^{i}{\Vert {\xi}_{t-i-1} \Vert}^{2}_{Q^{-1}}\right\}
\\
&\quad +\frac{\Delta}{1-\rho^{M_t}}
\end{aligned}
\end{equation}
Based on the fact that $\sqrt{a+b}\leq\sqrt{a}+\sqrt{b}$, and
\begin{equation}
\sum_{i=0}^{t-1}\sqrt{\rho}^{i}\leq\frac{1}{1-\sqrt{\rho}},     
\end{equation}
we can easily get \eqref{eq.stability results}.
\end{proof}

\bibliographystyle{ieeetr}
\bibliography{ref}
\end{document}